\documentclass[11pt]{article}

\usepackage{amsmath,amsthm,amssymb, enumerate}
\usepackage[utf8]{inputenc}
\usepackage[margin=3cm]{geometry}
\usepackage{charter}
\usepackage{tikz}
\usepackage[utf8]{inputenc} 
\usepackage[T1]{fontenc}    
\usepackage{hyperref}       
\hypersetup{
  colorlinks   = true, 
  urlcolor     = blue, 
  linkcolor    = blue, 
  citecolor   = red 
}
\usepackage{url}            
\usepackage{booktabs}       
\usepackage{amsfonts}       
\usepackage{nicefrac}       
\usepackage{microtype}      
\usepackage{amsfonts}  
\usepackage{graphicx}
\usepackage{fullpage}
\usepackage{braket}

\usepackage{graphicx, float}
\usepackage{subcaption}
\usepackage{url}
\usepackage{epsfig}
\usepackage{color}
\usepackage{refcount}
\usepackage{verbatim}
\usepackage{enumerate}
\usepackage{multirow}
\usepackage{framed}
\usepackage{spverbatim}
\usepackage[cal=rsfso,calscaled=.96]{mathalfa}
\usepackage{tikz}
\usepackage{color}
\definecolor{clemson-orange}{RGB}{234,106,32}
\definecolor{chicago-maroon}{RGB}{128,0,0}
\definecolor{cincinnati-red}{RGB}{190,0,0}
\definecolor{soft-cyan}{RGB}{68,85,90}
\usepackage{fullpage}
\usepackage{authblk}

\newcommand{\removed}[1]{}

\newcommand{\bb}{\mathbb}
\newcommand{\R}{\bb R}

\newcommand{\N}{{\bb N}}

\theoremstyle{definition}
\newtheorem{theorem}{Theorem}[section]

\newtheorem{lemma}[theorem]{Lemma}

\newtheorem{prop}[theorem]{Proposition}
\newtheorem{claim}{Claim}
\newtheorem{definition}{Definition}

\def\ve#1{\mathchoice{\mbox{\boldmath$\displaystyle\bf#1$}}
{\mbox{\boldmath$\textstyle\bf#1$}}
{\mbox{\boldmath$\scriptstyle\bf#1$}}
{\mbox{\boldmath$\scriptscriptstyle\bf#1$}}}

\newcommand{\x}{{\ve x}}
\newcommand{\y}{{\ve y}}
\newcommand{\z}{{\ve z}}

\renewcommand{\a}{{\ve a}}

\newcommand{\w}{{\ve w}}

\newcommand{\poly}{\textrm{poly}}

\newif\ifsolutions \solutionstrue

\title {Lower bounds over Boolean inputs for deep neural networks with ReLU gates.}

\author{ Anirbit Mukherjee\thanks{Department of Applied Mathematics and Statistics, Johns Hopkins University, Email: \tt{amukhe14@jhu.edu}}\ \ \ \ \ \ \ \ 
Amitabh Basu\thanks{Department of Applied Mathematics and Statistics, Johns Hopkins University, Email: \tt{basu.amitabh@jhu.edu}}
}

\date{}

\begin{document}

\maketitle

\begin{abstract}
~\\
Motivated by the resurgence of neural networks in being able to solve complex learning tasks we undertake a study of high depth networks using ReLU gates which implement the function $x \mapsto \max\{0,x\}$. We try to understand the role of depth in such neural networks by showing size lowerbounds against such network architectures in parameter regimes hitherto unexplored. In particular we show the following two main results about neural nets computing Boolean functions of input dimension $n$, 
\begin{itemize}
\item We use the method of random restrictions to show almost linear, $\Omega(\epsilon^{2(1-\delta)}n^{1-\delta})$, lower bound for completely weight unrestricted LTF-of-ReLU circuits to match the Andreev function on at least $\frac{1}{2} +\epsilon$ fraction of the inputs for $\epsilon > \sqrt{2\frac{\log^{\frac {2}{2-\delta}}(n)}{n}}$ for any $\delta \in (0,\frac 1 2)$
\item We use the method of sign-rank to show exponential in dimension lower bounds for ReLU circuits ending in a LTF gate and of depths upto $O(n^{\xi})$ with $\xi < \frac{1}{8}$ with some restrictions on the weights in the bottom most layer. All other weights in these circuits are kept unrestricted. This in turns also implies the same lowerbounds for LTF circuits with the same architecture and the same weight restrictions on their bottom most layer. 
\end{itemize}
Along the way we also show that there exists a $\mathbb{R}^ n\rightarrow \mathbb{R}$ Sum-of-ReLU-of-ReLU function which Sum-of-ReLU neural nets can never represent no matter how large they are allowed to be. 
\end{abstract}
\section{Introduction}

There has been a recent surge of activity in using neural networks for  complex artificial intelligence tasks (like this very recent spectacular demonstration ~\cite{silver2017mastering} of the power of neural nets). This has rekindled interest in understanding neural networks from a complexity theory perspective. A myriad of hard mathematical questions have surfaced in the attempts to rigorously explain the power of neural networks and a comprehensive overview of these can be found in this recent three part series of articles from The Center for Brains, Minds and Machines (CBMM), ~\cite{poggio2017and,poggio2017theory,zhang2017theory}.
 There is a rich literature investigating the complexity of the function classes represented by neural networks with various kinds of gates (or ``activation functions" which is the more common parlance in machine learning). Many papers, a canonical example being the classic paper by Maass~\cite{maass1997bounds}, establish complexity results for the entire class of functions represented by circuits where the gates can come from a very general family. This is complemented by papers that study a very specific family of gates such as the sigmoid gate or the LTF gate ~\cite{impagliazzo1997size}, ~\cite{siu1994rational}, ~\cite{sherstov2007powering} ~\cite{krause1994computational}, ~\cite{buhrman2007computation}, ~\cite{sherstov2009separating}, ~\cite{razborov2010sign}, ~\cite{bun2016improved}. Many associated results can also be found in these reviews ~\cite{lee2009lower, razborov1992small}. Recent circuit complexity results in~\cite{kane2016super},  ~\cite{tamaki2016satisfiability}, ~\cite{chen2016average}, ~\cite{kabanets2017polynomial} stand out as significant improvements over known lower (and upper) bounds on circuit complexity with threshold gates. The results of Maass~\cite{maass1997bounds} also show that very general families of neural networks can be converted into circuits with only LTF gates with at most a constant factor blow up in depth and polynomial blow up in size of the circuits.
~\\ \\
In the last 5 years or so, a particular family of gates called the {\em Rectified Linear Unit (ReLU)} gates have been reported to have significant advantages over more traditional gates in practical applications of neural networks. Such a gate with $n$ real inputs computes the following output, 
\begin{align}\label{eq:RELU-def}
\R^n &\rightarrow \R \\ 
\x &\mapsto \max \{0,b+\langle \w,\x \rangle\}
\end{align} 
~\\ 
where $\w \in \R^n$ and $b \in \R$ are fixed parameters associated with the gate ($b$ is called the bias of the gate). In comparison, the $\pm 1$ valued LTF gate mentioned above computes (for the same weights as above) the function, $(2\mathbf{1}_{(b+\langle \w,\x \rangle \geq 0)} -1)$ where $\mathbf{1}_{(b+\langle \w,\x \rangle \geq 0)}$ is the $0/1$ indicator function for the stated halfspace condition. 
~\\ \\
Some of the  prior results which apply to general gates, such as the ones in~\cite{maass1997bounds}, also apply to ReLU gates, because those results apply to gates that compute a piecewise polynomial function (ReLU is a piecewise linear function with only two pieces). However, as witnessed by results on LTF gates, one can usually make much stronger claims about specific classes of gates. To the best of our knowledge, no prior results have been obtained for ReLU gates from the perspective of {\em Boolean} complexity theory, i.e., the study of such circuits when restricted to Boolean inputs. The main focus of this work is to study circuits computing Boolean functions mapping $\{-1,1\}^m \rightarrow \{-1,1\}$ which use ReLU gates in their intermediate layers, and have an LTF gate at the output node (to ensure that the output is in $\{-1,1\}$). We remark that using an LTF gate at the output node while allowing more general analog gates in the intermediate nodes is a standard practice when studying the Boolean complexity of analog gates (see, for example,~\cite{maass1997bounds}).
~\\ \\
Although we are not aware of an analysis of lower bounds for ReLU circuits when applied to only Boolean inputs, there has been recent work on the analysis of such circuits when viewed as a function from $\R^n$ to $\R$ (i.e., allowing real inputs and output). From ~\cite{eldan2016power} and  ~\cite{daniely2017depth} (with restrictions on the domain and the weights) we know of (super-)exponential lowerbounds on the size of Sum-of-ReLU circuits for certain easy Sum-of-ReLU-of-ReLU functions . 
Depth v/s size tradeoffs 
for such circuits have recently also been studied in ~\cite{telgarsky2016benefits,hanin2017universal,liang2016deep, yarotsky2016error,safran2016depth} and in a recent paper ~\cite{arora2016understanding} by the current authors. To the best of our knowledge no lowerbounds scaling exponentially with the dimension are known for analog deep neural networks of depths more than $2$.
~\\ \\
In what follows, the {\em depth} of a circuit will be the length of the longest path from the output node to an input variable, and the {\em size} of a circuit will be the total number of gates in the circuit. We will also use the notation {\em Sum-of-ReLU} to refer to circuits whose inputs feed into a single layer of ReLU gates, whose outputs are combined into a weighted sum to give the final output. Similarly, Sum-of-ReLU-of-ReLU denotes the circuit with depth 3, where the output node is a simple weighted sum, and the intermediate gates are all ReLU gates in the two ``hidden" layers. We analogously define Sum-of-LTF, LTF-of-LTF, LTF-of-ReLU, LTF-of-LTF-of-LTF, LTF-of-ReLU-of-ReLU and so on. We will also use the notation LTF-of-(ReLU)$^k$ for a circuit of the form LTF-of-ReLU-of-RELU-$\ldots$-ReLU with $k\geq 1$ levels of ReLU gates. 

\section{Statement and discussion of results} 

\paragraph{Boolean v/s real inputs.} We begin our study with the following observation which shows that ReLU circuits have markedly different behaviour when the inputs are restricted to be Boolean, as opposed to arbitrary real inputs. Since AND and OR gates can both be implemented by ReLU gates, it follows that {\em any} Boolean function can be implemented by a ReLU-of-ReLU circuit. In fact, it is not hard to show something slightly stronger:

\begin{lemma}
Any function $f : \{-1,1\}^n \rightarrow \mathbb{R}$ can be implemented by a Sum-of-ReLU circuit using at most $\min\{2^n,\sum_{\hat{f}(S) \neq 0} \vert S \vert\}$ number of ReLU gates, where $\hat f(S)$ denotes the Fourier coefficient of $f$ for the set $S \subseteq \{1, \ldots, n\}$.
\end{lemma}
~\\ 
The Lemma follows by observing that the indicator functions of each vertex of the Boolean hypercube $\{-1,1\}^n$ can be implemented by a single ReLU gate, and the parity function on $k$ variables can be implemented by $k$ ReLU gates (see Appendix~\ref{sec:Parity-ReLU}). Thus, if one does not restrict the size of the circuit, then Sum-of-ReLU circuits can represent any pseudo-Boolean function. In contrast, we will now show that if one allows real inputs, then there exist functions with just 2 inputs (i.e., $n=2$) which cannot be represented by any Sum-of-ReLU circuit, no matter how large.

\begin{prop}\label{thm:max-0-x-y} The function $\max\{0,x_1,x_2\}$ cannot be computed by any Sum-of-ReLU circuit, no matter how many ReLU gates are used. It can be computed by a Sum-of-ReLU-of-ReLU circuit.
\end{prop}
~\\ 
The first part of the above proposition (the impossibility result) is proved in Appendix~\ref{sec:proof-max-0-x-y}. The second part follows from Corollary $2.2$ of a previous paper by the authors ~\cite{arora2016understanding}, which states that any $\R^n \to \R$ function that can be implemented by a circuit of ReLU gates, can always be implemented with at most $\lceil \log(n+1) \rceil$ layers of ReLU gates (with a weighted Sum to give the final output). 

\paragraph{Restricting to Boolean inputs.} From this point on, we will focus entirely on the situation where the inputs to the circuits are restricted to $\{-1,1\}$. One motivation behind our results is the desire to understand the strength of the ReLU gates vis-a-vis LTF gates. It is not hard to see that any circuit with LTF gates can be simulated by a circuit with ReLU gates with at most a constant blow-up in size (because a single LTF gate can be simulated by 2 ReLU gates when the inputs are a discrete set -- see Appendix~\ref{sec:LTF-ReLU}). The question is whether ReLU gates can do significantly better than LTF gates in terms of depth and/or size.
~\\ \\
A quick observation is that Sum-of-ReLU circuits can be linearly (in the dimension $n$) smaller than Sum-of-LTF circuits. More precisely, 

\begin{prop} The function $f:\{-1,1\}^n \to \R$ given by $f(x) =\sum_{i=1}^n 2^i \big(\frac{1+x_i}{2}\big)$ can be implemented by a Sum-of-ReLU circuit with 2 ReLU gates, and any Sum-of-LTF that implements $f$ needs $\Omega(n)$ gates.
\end{prop}
~\\
The above result follows from the following two facts: 1) any linear function is implementable by 2 ReLU gates, and 2) any Sum-of-LTF circuit with $w$ LTF gates gives a piecewise constant function that takes at most $2^w$ different values. Since $f$ takes $2^n$ different values (it evaluates every vertex of the Boolean hypercube to the corresponding natural number expressed in binary), we need $w \geq n$ gates.
~\\ \\
In the context of these preliminary results, we now state our main contributions. 
For the next result we recall the definition of the Andreev function ~\cite{andreev1987one} which has previously many times been used to prove computational lower bounds ~\cite{paterson1993shrinkage,impagliazzo1988decision,impagliazzo2012pseudorandomness}.

\begin{definition}[{\bf Andreev's function}]\label{Andreev}
The Andreev's function is the following mapping, 

\begin{align*}
A_n : \{0,1\}^{\lfloor \frac {n}{2} \rfloor} \times \{0,1\}^{ \lfloor \log (\frac{n}{2} ) \rfloor \times \lfloor \frac{n}{2\lfloor \log (\frac {n}{2}) \rfloor} \rfloor} &\longrightarrow \{0,1\} \\
(\x, [a_{ij}]) &\longmapsto x_{\text{bin} \left( \{ (\sum_{j=1}^{\lfloor \frac{n}{2 \lfloor \log (\frac {n}{2}) \rfloor} \rfloor} a_{ij}) \mod 2 \}_{i=1,2,..,\lfloor \log (\frac {n}{2} ) \rfloor } \right )}
\end{align*}
where ``bin" is the function that gives the decimal number that can be represented by its input bit string.
\end{definition}

~\\
We are particularly inspired by the most recent use of the Andreev function by Kane and Williams ~\cite{kane2016super} to get the first super linear lower bounds for approximating it using LTF-of-LTF circuits. We will give an almost linear lower bound on the size of LTF-of-ReLU circuits approximating this Andreev function with no restriction on the weights $\w, b$ for each gate.

\begin{theorem}\label{thm:andreev-LTF-ReLU}
For any $\delta \in (0,\frac{1}{2})$, there exists $N(\delta) \in \N$ such that for all $n \geq N(\delta)$ and $\epsilon > \sqrt{\frac{2\log^{\frac 2 {2-\delta}}(n)}{n}}$, any LFT-of-ReLU circuit on $n$ bits that matches the Andreev function on $n-$bits for at least $1/2 + \epsilon$ fraction of the inputs, has size $\Omega(\epsilon^{2(1-\delta)}n^{1-\delta})$. 
\end{theorem}

~\\ 
It is well known that proving lower bounds without restrictions on the weights is much more challenging even in the context of LTF circuits. In fact, the recent results in ~\cite{kane2016super} are the first superlinear lower bounds for LTF circuits with no restrictions on the weights. With restrictions on some or all the weights, e.g., assuming $poly(n)$ bounds on the weights (typically termed the ``small weight asssumption") in certain layers, exponential lower bounds have been established for LTF circuits ~\cite{hajnal1987threshold,impagliazzo1997size,sherstov2009separating,sherstov2011unbounded}. Our next results are of this flavor: under certain kinds of weight restrictions, we prove exponential size lower bounds on the size of LTF-of-(ReLU)$^{d-1}$ circuits. {\em One thing to note is that our weight restrictions are assumed only on the bottom layer (closest to the input). The other layers can have gates with unbounded weights.} Nevertheless, our weight restrictions are somewhat unconventional. 

\begin{definition}\label{def:weight-restriction}[Weight restriction condition] Let $m \in \N$ and $\sigma$ be any permutation of $\{1, \ldots, 2^m\}$. Let us also consider an arbitrary sequencing $\{\x^1, \ldots, \x^{2^m}\}$ of the vertices of the hypercube $\{-1,1\}^m$. Define the polyhedral cone $$P_{m,\sigma} := \{\a \in \R^m: \langle \a, \x^{\sigma(1)} \rangle \leq \langle \a, \x^{\sigma(2)} \rangle \leq \ldots \langle \a, \x^{\sigma(2^m)} \rangle\}.$$
In words, $P_{m,\sigma}$ is the set of all linear objectives that order the vertices of the $m$-dimensional hypercube in the order specified by $\sigma$. We will impose the condition that there exists a $\sigma$ such that for each ReLU gate in the bottom layer, the vector $\w \in P_{m,\sigma}$ ($\w$ as defined in~\eqref{eq:RELU-def}) and all weights are integers with magnitude bounded by some $W>0$.
\end{definition}
~\\ 
We will prove our lower bounds against the function proposed by Arkadev and Nikhil in ~\cite{chattopadhyay2017weights},
\begin{align}\label{AN}
g : OMB_n^0 \circ OR_{n^{\frac 1 3} - \log n } \circ XOR_2 : \{-1,1\}^{2(n^{\frac 4 3} -  n \log n)} \rightarrow \{-1,1\}
\end{align}
which we will refer to as the Arkadev-Nikhil function in the remainder of the paper. Here OMB is the ODD-MAX-BIT function which is a $\pm 1$ threshold gate which evaluates to $-1$ on say a $n-$bit input $\x$ if $\sum_{i=1}^n (-1)^{i+1}2^i(1+x_i) \geq \frac{1}{2}$. We show the following exponential lowerbound against this function, 
~\\ 
\begin{theorem}\label{deepReLU}
Let $m, d, W\in \N$. Any depth $d$ LTF-of-(ReLU)$^{d-1}$ circuits on $2m$ bits such that the weights in the bottom layer are restricted as per Definition~\ref{def:weight-restriction} that implements the Arkadev-Nikhil function on $2m$ bits will require a circuit size of $$\Omega \left ( (d-1) \frac{2^{\frac {m^{\frac 1 8}}{ d-1}}}{(mW)^{\frac 1 {d-1}}}\right ).$$
Consequently, one obtains the same size lower bounds for circuits with only LTF gates of depth $d$.
\end{theorem}
~\\
Note that this is an exponential in dimension size lowerbound for even super-polynomially growing bottom layer weights (and additional constraints as per Definition~\ref{def:weight-restriction}) and upto depths scaling as $d = O(m^\xi)$ with $\xi < \frac{1}{8}$.
~\\ \\
We note that the Arkadev-Nikhil function can be represented by an $O(m)$ size LTF-of-LTF circuit with no restrictions on weights (see Theorem~\ref{thm:arkadev-nikhil} below). In light of this fact, Theorem~\ref{deepReLU} is somewhat surprising as it shows that for the purpose of representing Boolean functions a deep ReLU circuit (ending in a LTF) gate can get exponentially weakened when just its bottom layer weights are restricted as per Definition~\ref{def:weight-restriction}, even if the integers are allowed to be super-polynomially large. Moreover, the lower bounds also hold of LTF circuits of arbitrary depth $d$, under the same weight restrictions on the bottom layer. We are unaware of any exponential lower bounds on LTF circuits of arbitrary depth under any kind of weight restrictions.


~\\ 
We will use the method of sign-rank to obtain the exponential lowerbounds in Theorems~\ref{deepReLU}.
The {\emph sign-rank} of a real matrix $A$ with all non-zero entries is the least rank of a matrix $B$ of the same dimension with all non-zero entries such that for each entry $(i,j)$, $sign(B_{ij}) = sign(A_{ij})$. For a Boolean function $f$ mapping, $f : \{-1,1\}^m \times \{-1,1\}^m \rightarrow \{-1,1\}$ one defines the ``sign-rank of f" as the sign-rank of the $2^m \times 2^m$ dimensional matrix $[f(\x,\y)]_{\x ,\y \in \{-1,1\}^m}$. This notion of a sign-rank has been used to great effect in diverse fields from communication complexity to circuit complexity to learning theory. Explicit matrices with a high sign-rank were not known till the breakthrough work by Forster, \cite{forster2002linear}. Forster et. al. showed elegant use of this complexity measure to show exponential lowerbounds against LTF-of-MAJ circuits in \cite{forster2001relations}. Lot of the previous literature about sign-rank has been reviewed in the book by Satya Lokam~\cite{lokam2009complexity}. Most recently the following result was obtained by Arkadev and Nikhil in~\cite{chattopadhyay2017weights} leading to a proof of strict containment of LTF-of-MAJ in LTF-of-LTF. 

\begin{theorem}\label{thm:arkadev-nikhil}[{\bf Theorem $4.2$ and Corollary $1.2$ in ~\cite{chattopadhyay2017weights}}]
~\\
The Akradev-Nikhil function $g$ in equation \ref{AN} can be represented by a linear sized LTF-of-LTF circuit and $\text{sign-rank}(g) \geq \frac {2^{n^{\frac 1 3} - 2 \log n } }{16} $
\end{theorem}
~\\
We will prove our theorem by showing a small upper bound on the sign-rank of LTF-of-(ReLU)$^{d-1}$ circuits which have their bottom most layer's weight restricted in the said way.

\section{Lower bounds for LTF-of-ReLU against the Andreev function (Proof of Theorem~\ref{thm:andreev-LTF-ReLU})}

We will use the classic ``method of random restrictions" ~\cite{subbotovskaya1961realizations,stad1998shrinkage,hastad1986almost,yao1985separating,rossman2008constant} to show a lowerbound for weight unrestricted LTF-of-ReLU circuits for representing the Andreev function. The basic philosophy of this method is to take any arbitrary LTF-of-ReLU circuit which supposedly matches the Andreev function on a large fraction of the inputs and to randomly fix the values on some of its input coordinates and also do the same fixing on the same coordinates of the input to the Andreev function. Then we show that upon doing this restriction the Andreev function collapses to an arbitrary Boolean function on the remaining inputs (what it collapses to depends on what values were fixed on its inputs that got restricted).  But on the other hand we show that the LTF-of-ReLU collapses to a circuit which is of such a small size that with high-probability it cannot possibly approximate a randomly chosen Boolean function on the remaining inputs. This contradiction leads to a lowerbound. 
~\\ \\
There are two important concepts towards implementing the above idea. First one has to precisely define as to when can a ReLU gate upon a partial restriction of its inputs be considered to be removable from the circuit. Once this notion is clarified it will automatically turn out that doing random restrictions on ReLU is the same as doing random restriction on a LTF gate as was recently done in ~\cite{kane2016super}. The secondly it needs to be true that at any fixed size LTF-of-ReLU circuits cannot represent too many of all the Boolean functions possible at the same input dimension.  For this very specific case of LTF-of-ReLU circuits where ReLU gates necessarily have a fan-out of $1$, Theorem 2.1 in ~\cite{maass1997bounds} applies and we have from there that LTF-of-ReLU circuits over $n-$bits with $w$ ReLU gates can represent at most $N = 2^{O((wn + w + w+1 +1)^2\log (wn + w + w+1 +1))} = 2^{O((wn+2w+2)^2\log (wn+2w+2))}$ number of Boolean functions. We note that slightly departing from the usual convention with neural networks here in this work by Wolfgaang Mass he allows for direct wires from the input nodes to the output LTF gate. This flexibility ties in nicely with how we want to define a ReLU gate to be becoming removable under the random restrictions that we use.
~\\ 



\paragraph{Random Boolean functions vs any circuit class}
~\\ \\
In everything that follows all samplings being done (denoted as $\sim$) are to be understood as sampling from an uniform distribution unless otherwise specified. Firstly we note this well-known lemma, 

\begin{claim}
Let $f : \{-1,1\}^n \rightarrow \{-1,1\}$ be any given Boolean function. Then the following is true, 
\[ \mathbb{P}_{g \sim \{ \{-1,1\}^n \rightarrow \{-1,1\} \}} \left [ \mathbb{P}_{\x \sim \{-1,1\}^n} [ f(\x) = g(\x) ] \geq \frac {1}{2} + \epsilon \right ] \leq e^{-2^{n+1}\epsilon^2} \]
\end{claim}

~\\
From the above it follows that if $N$ is the total number of functions in any circuit class (whose members be called $C$) then we have by union bound,

\begin{align}\label{randomapprox}
\mathbb{P}_{g \sim \{ \{-1,1\}^n \rightarrow \{-1,1\} \}} \left [ \exists C \text{ s.t } \mathbb{P}_{\x \sim \{-1,1\}^n} [ C(\x) = g(\x) ] \geq \frac {1}{2} + \epsilon \right ] \leq N e^{-2^{n+1}\epsilon^2}
\end{align}
~\\ 
~\\ 
Equipped with these basics we are now ready to begin the proof of the lowerbound against weight unrestricted LTF-of-ReLU circuits, 
\begin{proof}
\begin{definition}
Let $D$ denote arbitrary LTF-of-ReLU circuits over $\lfloor \log (\frac{n}{2} ) \rfloor$ bits.
\end{definition}
 
~\\ 
 For some $\frac{\epsilon}{3} \leq \frac{1}{2}$ and a size function denoted as $s(n,\epsilon)$ we use equation \ref{randomapprox} , the definition of $D$ above and the upperbound given earlier for the number of LTF-of-ReLU functions at a fixed circuit size (now used for circuits on $\lfloor \log (\frac{n}{2} ) \rfloor$  bits) to get, 
\begin{align*}
\mathbb{P}_{\substack{f \sim \{0,1\}^{\lfloor \log (\frac{n}{2} ) \rfloor} \rightarrow \{0,1\} }} & \Bigg [\forall D \text{ s.t } \vert D \vert \leq s(n,\epsilon)\text{ }\vert \mathbb{P}_{ \y \sim \{0,1\}^{\lfloor \log (\frac{n}{2} ) \rfloor}} [ f(\y) = D(\y) ]  \leq \Big (\frac {1}{2} + \frac {\epsilon}{3} \Big ) \Bigg ]\\
&\geq 1 - 2^{O(s^2\log^2(\frac{n}{2})\log(\log(\frac{n}{2})s))}e^{-\left (\frac{\epsilon^2}{9} \right)2^{1+\lfloor \log \left ( \frac{n}{2} \right ) \rfloor}}\\
&\geq 1 - 2^{O(s^2k^2\log(ks))}e^{-\left (\frac{2\epsilon^2}{9} \right)2^{k}}
\end{align*}
whereby in the last inequality above we have assumed that $n=2^{k+1}$. This assumption is legitimate because we want to estimate certain large $n$ asymptotics. For any arbitrarily chosen constant $C < \frac{2}{9}$ we try to satisfy the following condition, $O(s^2k^2\log(ks))-\frac{2\epsilon^2 2^k}{9} \leq -C\epsilon^22^k \implies O(s^2k^2\log(ks)) \leq O ( \epsilon^2 2^k) $. For any constant $\theta >0$ for large enough $x >0$ we would have $\log(x) < x^\theta$ and hence the above constraint on $s$ gets satisfied if we work in the regime, $s \leq O(\frac{(\epsilon^22^k)^{\frac{1}{2+\theta}}}{k})$. So for this range of $s$ we have, $2^{O(s^2k^2\log(ks))}e^{-\left (\frac{2\epsilon^2}{9} \right)2^{k}} \leq e^{O(s^2k^2\log(ks))-\left (\frac{2\epsilon^2}{9} \right)2^{k}} \leq e^{-C\epsilon^22^k}$. Now we want, $e^{-C\epsilon^22^k} \leq \frac{\epsilon}{3}$. But on the otherhand for the upperbound on $s$ to make sense we need, $\epsilon^2 2^k  \geq k^{2+\theta}$. Its clear that both the conditions get satisfied if for asymptotically large $n$ we choose $\epsilon > \sqrt{\frac{2\log^{2+\theta} (\frac{n}{2})}{n}}$. And corresponding to this we have for $s(n,\epsilon) \leq O(\frac{\epsilon^{\frac{2}{2+\theta}}n^{\frac{1}{2+\theta}}}{2^{\frac{1}{2+\theta}}\log (\frac{n}{2})})$
\begin{align}\label{smaass}
\nonumber \mathbb{P}_{\substack{f \sim \{0,1\}^{\lfloor \log (\frac{n}{2} ) \rfloor} \rightarrow \{0,1\} }} & \Bigg [\forall D \text{ s.t } \vert D \vert \leq s(n,\epsilon)\text{ }\vert \mathbb{P}_{ \y \sim \{0,1\}^{\lfloor \log (\frac{n}{2} ) \rfloor}} [ f(\y) = D(\y) ]  \leq \Big (\frac {1}{2} + \frac {\epsilon}{3} \Big ) \Bigg ]\\
&\geq  1 - \frac{\epsilon}{3} 
\end{align}
~\\
\begin{definition}[$\mathbf{ F^*}$]
Let $F^*$ be the subset of all these $f$ above for which the above event is true.\\
\end{definition}
~\\
Now we recall the definition of the Andreev function in equation \ref{Andreev} for the following definition and the claim, 
\begin{definition}[$\mathbf{\rho}$]
Let $\rho$ denote the set of all possible ``random restrictions" where one is fixing all the input bits of $A_n$ except $1$ bit in each row of the matrix $a$. So the restricted function (call it $A_n \vert _{\rho}$ by overloading the notation for simplicity) computes a function of the form,

\[ A_n \vert_{\rho} : \{0,1\}^{ \lfloor \log (\frac{n}{2} ) \rfloor } \rightarrow \{0,1\}\]
\end{definition}
~\\
From the definitions of $A_n$ and $\rho$ above the following is immediate, 

\begin{claim}
The truth table of $A_n \vert _\rho$ is the $\x$ string in the input to $A_n$ that gets fixed by $\rho$. Thus we observe that if $\rho$ is chosen uniformly at random then $A_n \vert _\rho$ is a $\lfloor \log (\frac{n}{2} ) \rfloor$ bit Boolean function chosen uniformly at random. 
\end{claim}
~\\ \\ 
Let $f^*$ be any arbitrary member of $F^*$. Let $\x^* \in \{0,1\}^{\lfloor \frac {n}{2} \rfloor}$ be the truth-table of $f^*$. Let $\rho(\x^*)$ be restrictions on the input of $A_n$ which fix the $\x$ part of its input to $\x^*$. So when we are sampling restrictions uniformly at random from the restrictions of the type $\rho(\x^*)$ these different instances differ in which bit of each row of the matrix $a$ (of the input to $A_n$) they left unfixed and to what values did they fix the other entries of $a$. Let $C$ be a $n$ bit LTF-of-ReLU Boolean circuit of size say $w(n,\epsilon)$. Thus under the restriction $\rho(\x^*)$ both $C$ and $A_n$ are $\lfloor \log (\frac{n}{2} ) \rfloor$ bit Boolean functions.
~\\ \\
Now we note that a ReLU gate over $n$ bits upon a random restriction becomes redundant (and hence removable) iff its linear argument either reduces to a non-positive definite function or a positive definite function. In the former case the gate is computing the constant function zero and in the later case it is computing a linear function which can be simply implemented by introducing wires connecting the inputs directly to the output LTF gate. Thus in both the cases the resultant function no more needs the ReLU gate for it to be computed.  (We note that such direct wires from the input to the output gate were allowed in how the counting was done of the total number of LTF-of-ReLU Boolean functions at a fixed circuit size.) Combining both the cases we note that the conditions for collapse (in this sense) of a ReLU gate is identical to that of the conditions of collapse for a LTF gate with the same linear argument. Hence corresponding to the random restrictions $\rho$ we can just directly utilize the random restriction lemma $1.1$ from ~\cite{kane2016super} to say that, 
\[ \mathbb{P}_{\rho(\x^*)} [  \text{ReLU}\vert_{\rho(\x^*)} \text {is removable }] \geq \eta \]
where for $\eta = 1 - O(\frac{\log n}{\sqrt{n}})$
\newline
The above definition of $\eta$ implies, 
\begin{align}\label{useofeta}
\nonumber &\mathbb{P}_{\rho(\x^*)} [ \text { A $n-$bit  \text{ReLU} is \emph {not} forced to a constant }] \leq 1 - \eta\\
\nonumber \implies &\mathbb{E}_{\rho(\x^*)} [ \text { Number of  \text{ReLU}s of C \emph {not} forced to a constant }] \leq w(n,\epsilon)(1 - \eta)\\
\nonumber \implies &\mathbb{P}_{\rho(\x^*)} [\text { Number of  \text{ReLU}s of C \emph {not} forced to a constant } > s(n,\epsilon) ]\\
\nonumber &\leq \frac {\mathbb{E}_{\rho(\x^*)} [ \text { Number of  \text{ReLU}s of C \emph {not} forced to a constant }]}{s(n,\epsilon)}\\
\nonumber \implies &\mathbb{P}_{\rho(\x^*)} [\text { Number of  \text{ReLU}s of C \emph {not} forced to a constant } \geq s(n,\epsilon) ] \leq \frac {w(n,\epsilon)(1-\eta)}{s(n,\epsilon)}\\
\nonumber \implies &\mathbb{P}_{\rho(\x^*)} [\text { Size of } C \vert _{\rho(\x^*)} \leq s(n,\epsilon) ] \geq 1 - \frac {w(n,\epsilon)(1-\eta)}{s(n,\epsilon)}\\
\end{align}
~\\
Now we compare with the definitions of $\epsilon$ and $f^*$ to observe that (a) with probability at least $1 - \frac {w(n,\epsilon)(1-\eta)}{s(n,\epsilon)}$, $C \vert _{\rho(\x^*)}$ is of the circuit type as in the event in equation $2$ and (b) by definition of the Andreev function it follows that $A_n\vert_{\rho(\x^*)}$ has its truth table given by $\x^*$ and hence it specifies the same function as $f^* \in F^*$. Hence $\forall \x^* \text{ and } \rho(\x^*)$ this can as well write this as, 

\begin{align}\label{Anisrandom}
\mathbb{P}_{\y \sim \{0,1\}^{\lfloor \log (\frac{n}{2} ) \rfloor}} [C\vert _{\rho(\x^*)}(\y) = A_n \vert _{\rho(\x^*)}(\y) \vert  \text{ Size of } C\vert _{\rho(\x^*)} \leq s(n,\epsilon) ] \leq \frac {1}{2} + \frac {\epsilon}{3}
\end{align}
~\\
$\forall \x^*$ equation \ref{useofeta} can be rewritten as, 
\begin{align}\label{useofeta2}
\mathbb{P}_{\rho(\x^*)} [\text { Size of } C \vert _{\rho(\x^*)} \leq s(n,\epsilon) ] \geq 1 - \frac {w(n,\epsilon)(1-\eta)}{s(n,\epsilon)} 
\end{align}
The equation \ref{smaass} can be written as, 
\begin{align}\label{useofs}
\mathbb{P}_{f \sim \{0,1\}^{\lfloor \log (\frac{n}{2} ) \rfloor} \rightarrow \{0,1\} } [ f \in F^* ]  \geq 1- \frac {\epsilon}{3}
\end{align}

\begin{claim}{Circuits $C$ have low correlation with the Andreev function}
\begin{align*}
\mathbb{P}_{\z \sim \{0,1\}^n}[C(z) = A_n(z)] \leq \frac {\epsilon}{3} + \frac {w(n,\epsilon)(1-\eta)}{s(n,\epsilon)} + \frac {1}{2} + \frac {\epsilon}{3}\\
\end{align*}
\end{claim}
\begin{proof}
We think of sampling a $z \sim \{0,1\}^n$ as a two step process of first sampling a $\tilde{f}$, a $\lfloor \log (\frac n 2) \rfloor$ bit Boolean function and fixing the first $\lfloor \frac n 2 \rfloor$ bits of $z$ to be the truth-table of $\tilde{f}$ and then we randomly assign values to the remaining $\lfloor \frac n 2 \rfloor$ bits of $z$. Call these later $\lfloor \frac n 2 \rfloor$ bit string to be $\x_{other}$. 
\begin{align*}
\mathbb{P}_{\z \sim \{0,1\}^n}[C(z) = A_n(z)] &= \mathbb{E}_{\z \sim \{0,1\}^n}[\mathfrak{1}_{C(z) = A_n(z)} ]\\
&=\mathbb{E}_{\z \sim \{0,1\}^n}[\mathfrak{1}_{C(z) = A_n(z)} \mathfrak{1}_{\tilde{f} \in F^*} ] + \mathbb{E}_{\z \sim \{0,1\}^n}[\mathfrak{1}_{C(z) = A_n(z)} \mathfrak{1}_{\tilde{f} \notin F^*}]\\
&= \mathbb{P}_{\z \sim \{0,1\}^n}[(C(z) = A_n(z))\cap(\tilde{f} \in F^*)]+ \mathbb{P}_{\z \sim \{0,1\}^n}[(C(z) = A_n(z))\cap(\tilde{f} \notin F^*)]\\
&= \mathbb{P}_{\z \sim \{0,1\}^n}[(C(z) = A_n(z))\mid (\tilde{f} \in F^*)]\mathbb{P}_{\z \sim \{0,1\}^n} [\tilde{f} \in F^* ]\\
&+ \mathbb{P}_{\z \sim \{0,1\}^n}[(C(z) = A_n(z))\cap(\tilde{f} \notin F^*)]\\
&\leq \mathbb{P}_{\z \sim \{0,1\}^n}[(C(z) = A_n(z))\mid (\tilde{f} \in F^*)] + \mathbb{P}_{\z \sim \{0,1\}^n}[\tilde{f} \notin F^*]\\
&\leq \mathbb{P}_{\z \sim \{0,1\}^n}[(C(z) = A_n(z))\mid (\tilde{f} \in F^*)] + \frac \epsilon 3\\
\end{align*}
In the last line above we have invoked equation \ref{useofs}. Now we note that sampling the $n$ bit string $z$ such that $\tilde{f} \in F^*$ is the same as doing a random restriction of the type $\rho(\tilde{f})$ and then randomly picking a  $\lfloor \log (\frac n 2) \rfloor$ bit string say $\y$. So we can rewrite the last inequality as, 

\begin{align*}
\mathbb{P}_{\z \sim \{0,1\}^n}[C(z) = A_n(z)]&\leq \mathbb{P}_{(\rho(\tilde{f}),\y)}[C(\rho(\tilde{f}),\y) = A_n(\rho(\tilde{f}),\y)] + \frac \epsilon 3\\
&\leq \mathbb{E}_{(\rho(\tilde{f}),\y)}[\mathfrak{1}_{C(\rho(\tilde{f}),\y) = A_n(\rho(\tilde{f}),\y)} \mid (\tilde{f} \in F^*)] + \frac \epsilon 3\\
&\leq \mathbb{E}_{(\rho(\tilde{f}),\y)}[\mathfrak{1}_{C(\rho(\tilde{f}),\y) = A_n(\rho(\tilde{f}),\y)}\mathfrak{1}_{\text{Size of } C\vert_{\rho(\tilde{f})} < s(n,\epsilon)} \mid (\tilde{f} \in F^*)]\\
&+ \mathbb{E}_{(\rho(\tilde{f}),\y)}[\mathfrak{1}_{C(\rho(\tilde{f}),\y) = A_n(\rho(\tilde{f}),\y)}\mathfrak{1}_{\text{Size of } C\vert_{\rho(\tilde{f})} \geq s(n,\epsilon)}\mid (\tilde{f} \in F^*)] + \frac \epsilon 3\\
&\leq \mathbb{P}_{(\rho(\tilde{f}),\y)}[C(\rho(\tilde{f}),\y) = A_n(\rho(\tilde{f}),\y) \mid \left ( (\text{Size of } C\vert_{\rho(\tilde{f})} < s(n,\epsilon)) \cap  (\tilde{f} \in F^*) \right )]\\
&+\mathbb{P}_{(\rho(\tilde{f}),\y)}[\text{Size of} C\vert_{\rho(\tilde{f})} \geq s(n,\epsilon)\mid (\tilde{f} \in F^*)] + \frac \epsilon 3\\
&\leq \left ( \frac 1 2 + \frac \epsilon 3\right ) + \frac {w(n,\epsilon)(1-\eta)}{s(n,\epsilon)} + \frac \epsilon 3
\end{align*}
In the last step above we have used equations \ref{Anisrandom} and \ref{useofeta2}. 
\end{proof}
~\\
So after putting back the values of $\eta$ and the largest scaling of $s(n,\epsilon)$ that we can have (from equation \ref{smaass}), the upperbound on the above probability becomes, 
\begin{align*}
\frac{1}{2} + \frac{2\epsilon}{3} + O \Bigg ( \frac{w(n,\epsilon)\log (n) }{\sqrt{n} (\frac{\epsilon^{\frac{2}{2+\theta}}n^{\frac{1}{2+\theta}}}{2^{\frac{1}{2+\theta}}\log (\frac{n}{2})})} \Bigg ) 
\end{align*}
~\\
Thus the probability is upperbounded by $\frac{1}{2} + \epsilon$ as long as $w(n,\epsilon) = O \Bigg ( \frac{\epsilon^{1+\frac{2}{2+\theta}} n^{\frac{1}{2} + \frac{1}{2+\theta}} \log \Big ( \frac{n}{2} \Big ) }{\log (n) }\Bigg )$

~\\ 
Stated as a lowerbound we have that if a LTF-of-ReLU has to match the $n-$bit Andreev function on more than $\frac{1}{2}+\epsilon$ fraction of the inputs for $\epsilon > \sqrt{\frac{2\log^{2+\theta} (\frac{n}{2})}{n}}$ for some $\theta >0$ (asymptotically this is like having a constant $\epsilon$) then the LTF-of-ReLU needs to be of size $\Omega(\epsilon^{\frac{4+\theta}{2+\theta}}n^{\frac{1}{2} + \frac{1}{2+\theta}})$. Now we define $\delta \in (0,\frac 1 2)$ such that $\delta = \frac {\theta}{2(2+\theta)}$ and that gives the form of the almost linear lowerbound as stated in the theorem. 
\end{proof}

\section{Smaller upper bounds on the sign-rank of LTF-of-(ReLU)$^{\mathbf d-1}$ with weight restrictions only on the bottom most layer (Proof of Theorem \ref{deepReLU})}

For a $\{-1,1\}^M \rightarrow \{-1,1\}$ LTF-of-ReLU circuit with any given weights on the network the inputs to the threshold function of the top LTF gate are some set of $2^M$ real numbers (one for each input). Over all these inputs let $p>0$ be the distance from $0$ of the largest negative number on which the LTF gate ever gets evaluated. Then by increasing the bias at this last LTF gate by a quantity less then $p$ we can ensure that no input to this LTF gate is $0$ while the entire circuit still computes the same Boolean function as originally. So we can assume without loss of generality that the input to the threshold function at the top LTF gate is never $0$. We also recall that the weights at the bottom most layer are constrained to be integers of magnitude at most $W>0$. 

~\\ 
Let this depth $d$ LTF-of-(ReLU)$^{d-1}$ circuit map $\{-1,1\}^{m}\times\{-1,1\}^{m} \rightarrow \{-1,1\}$. Let $\{w_{k}\}_{k=1}^{d-1}$ be the widths of the ReLU layers at depths indexed by increasing $k$ with increasing distance from the input. Thus, the output LTF gate gets $w_{d-1}$ inputs; the $j$-th input, for $j=1,2,..,w_{d-1}$, is the output of a circuit $C_j$ of depth $d-1$ composed of only ReLU gates. Let $f_j(\x,\y):\{-1,1\}^{m}\times\{-1,1\}^{m} \rightarrow \R$ be the pseudo-Boolean function implemented by $C_j$. 
Thus the output of the overall LTF-of-(ReLU)$^{d-1}$ circuit is, 
\begin{align}\label{top} 
 f(\x,\y):=\text{LTF} \left [ \beta + \sum_{j=1}^{w_{d-1}} \alpha_j f_j(\x,\y)\right]
\end{align}

\begin{lemma}\label{ReLU-circuit-rank}
Let $k \geq 0$ and $w_1, \ldots, w_k \geq 1$ be natural numbers. Consider a circuit with $2m$ inputs and a single output, consisting of only ReLU gates of depth $k+1$ with $w_i$ ReLU gates at each depth, with $i=1$ corresponding to the layer closest to the input (note that single output ReLU gate is not counted here). We restrict the inputs to $\{-1,1\}^m \times \{-1,1\}^m$, so the circuit implements a pseudo-Boolean function $g: \{-1,1\}^m \times \{-1,1\}^m\to \R$. 
~\\
Assume that the weights of the $w_1$ ReLU gates in the layer closest to the input are restricted as per Definition~\ref{def:weight-restriction}. Define the $2^m\times 2^m$ matrix $G(\x,\y)$ whose rows and columns are indexed by $(\x,\y) \in \{-1,1\}^m \times \{-1,1\}^m$ as $$G(\x,\y) = g(\x,\y).$$ Then $G$ has a block structure, where the rows and columns can be partitioned {\em contiguously} into $O\big((\prod_{i=1}^k w_i)(mW)\big)$ blocks (thus, $G$ has $O\big((\prod_{i=1}^k w_i)^2(mW)^2\big)$ blocks), and within each block $G$ is constant valued.
\end{lemma}
~\\
Before we prove the Lemma, let us see why it implies Theorem~\ref{deepReLU}. Let $F_j(\x,\y)$ be the matrix obtained from the ReLU circuit outputs $f_j(\x,\y)$ from~\eqref{top}, and let $F(\x,\y)$ be the matrix obtained from $f(\x,\y)$. Let $J_{2^m\times 2^m}$ be the matrix of all ones. Then

\begin{align*}
\text{sign-rank}(F(\x,\y))
= &  \;\text{sign-rank} \left ( \text{sign} \left [ \beta J_{2^m\times 2^m} + \sum_{j=1}^{w_{d-1}} \alpha_j F_j(\x,\y) \right] \right )\\
\leq &\; \text{rank} \left ( \beta J_{2^m\times 2^m} + \sum_{j=1}^{w_{d-1}} \alpha_j F_j(\x,\y) \right )\\
\leq &\; 1 + \sum_{j=1}^{w_{d-1}}\text{rank}(F_j(\x,\y))\\
= & O\left (\left ( \prod_{k=1}^{d-1} w_k\right )^2 (mW)^2\right )
\end{align*}
where the first inequality follows from the definition of sign-rank, the second inequality follows from the subadditivty of rank and the last inequality is a consequence of Lemma~\ref{ReLU-circuit-rank}. Indeed, a matrix with block structure as in the conclusion of Lemma~\ref{ReLU-circuit-rank} has rank at most $O\big((\prod_{i=1}^k w_i)^2(mW)^2\big)$ by expressing it as a sum of these many matrices of rank one and using subaddivity of rank.

~\\
Now we recall that the Arkadev-Nikhil function $g$ (which is linear sized depth $2$ LTF) on $2m = 2(n^{\frac{4}{3}}-n\log n)$ bits has sign-rank $\Omega(2^{n^{\frac{1}{3} - 2\log n}})$. It follows that $n^{\frac{4}{3}} \geq m$ and for any constant $C$ s.t $C \in (0,1)$ for large enough $n$ we would have, $\text{sign-rank}(g) = \Omega(2^{Cn^{\frac{1}{3}}}) = \Omega(2^{Cm^{\frac{1}{4}}})$. From the above upper bound on the sign-rank of our bottom layer weight restricted LTF-of-(ReLU)$^{d-1}$ with widths $\{w_k\}_{k=1}^{d-1}$ it follows that for this to represent this Arkadev-Nikhil function it would need, $\left (\left ( \prod_{k=1}^{d-1} w_k\right )^2 (mW)^2\right ) = \Omega(m^{\frac{1}{4}})$. Hence it follows that the size ($1+\sum_{k=1}^{d-1}w_i$) required for such  LTF-of-(ReLU)$^{d-1}$ circuits to represent the Arkadev-Nikhil function is $\Omega \left ( (d-1) \frac{2^{\frac {m^{\frac 1 8}}{ d-1}}}{(mW)^{\frac 1 {d-1}}}\right )$.
~\\
The statement about LTF circuits is a straightforward consequence of the above result and Claim~\ref{LLbyLR} in Appendix~\ref{sec:LTF-ReLU} which says that any LTF gate can be simulated by 2 ReLU gates.
~\\
We now prove Lemma~\ref{ReLU-circuit-rank}.

\begin{proof}[Proof of Lemma~\ref{ReLU-circuit-rank}]
We will prove this Lemma by induction on $k$. 

\paragraph{The base case of the induction $k=0$: A single ReLU gate.} A single ReLU gate's output is given by $\max\{0, \langle \a^1, \x\rangle + \langle \a^2,\y\rangle + b\}$, where $\a^1, \a^2 \in \R^m$ and $b \in \R$. Since the entries of $\a^1, \a^2$ and $b$ are assumed to be integers bounded by $W >0$, the terms $\langle \a^1, \x\rangle$ and $\langle \a^2,\y\rangle$ can each take at most $O(mW)$ different values, since $\x, \y \in \{-1,1\}^m$. So we can arrange the rows and columns in increasing order of $\langle \a^1,\x\rangle$ and $\langle \a^2,\y\rangle$ and then partition the rows and columns contiguously according to these values, and the base case is proved.

\paragraph{The induction step.} We first make a simple claim about the sum of matrices which are block wise constant.

\begin{claim}\label{claim:block-matrices}
Let $w, M, D$ be fixed natural numbers. Let $A_1, \ldots, A_w$ be any $M\times M$ matrices such that for each $A_i$ the rows and columns can be partitioned contiguously into $D$ blocks (not necessarily equal in size), such that $A_i$ is constant valued within each of the $D^2$ blocks. Then $A := A_1 + \ldots + A_w$ is an $M\times M$ matrix whose rows and columns can be partitioned contiguously into $w(D-1) + 1$ blocks such that $A$ is constant valued within each block defined by this partition of the rows and columns.
\end{claim}

\begin{proof} The partition of the rows of $A_i$ into $D$ contiguous blocks is equivalent to a choice of $D-1$ lines out of $M-1$ lines. When we sum the matrices, the refined partition in the sum is a selection of $w(D-1)$ lines out of $M-1$ lines, giving us $w(D-1) + 1$ contiguous blocks. The same argument holds for the columns.
\end{proof}
~\\
To complete the induction step, we observe that a ReLU circuit with depth $k+1$ layers can be seen as computing $g(\x, \y) = \max\{0, b + \sum_{i=1}^{w_{k}} a_jg_i(\x,\y)\},$ where $g_i(\x,\y)$ is the output of a ReLU circuit of depth $k$. Thus, the corresponding matrices satisfy $G(\x, \y) = \max\{0, bJ_{2^m\times 2^m} + \sum_{i=1}^{w_{k}} a_jG_i(\x,\y)\}$, where $J_{2^m\times 2^m}$ is the matrix of all ones, and the ``max'' is taken entrywise. the induction hypothesis tells us that the rows and columns each matrix $G_i$ can be partitioned contiguously into $O\big((\prod_{i=1}^{k-1} w_i)(mW)\big)$ such that $G_i$ is constant valued within each block. Thus, by Claim~\ref{claim:block-matrices}, the rows and columns of the matrix $bJ_{2^m\times 2^m} + \sum_{i=1}^{w_{k}} a_jG_i(\x,\y)$ can be partitioned into $O\big((\prod_{i=1}^k w_i)(mW)\big)$ contiguous blocks.
\end{proof}
\section{Acknowledgements}
We would like to thank Aurko Roy (Google Brain, San Francisco Bay Area) for extensive discussions on the methods used and the questions addressed in this work. We also thank Nikhil Mande (TIFR), Piyush Srivastava (TIFR) and Xin Li (JHU) for helpful conversations on circuit complexity. Amitabh Basu and Anirbit Mukherjee gratefully acknowledge support from the NSF grant CMMI1452820.

\bibliographystyle{abbrv}
\bibliography{references}

\appendix 

\section{Proof of Proposition~\ref{thm:max-0-x-y}}\label{sec:proof-max-0-x-y}

We first observe that the set of points where $\max\{0,x_1,x_2\}$ is not differentiable is precisely the union of the three half-lines (or rays) $\{(x_1,x_2): x_1=x_2, x_1\geq 0\} \cup \{(0,x_2): x_2 \leq 0\} \cup \{(x_1,0): x_1 \leq 0\}$. 
On the other hand, consider any Sum-of-ReLU circuit, which can be expressed as a function of the form
$$f(x) = \sum_{i=1}^w c_i\max\{0,\langle a^i, x \rangle + b_i\},$$ where $w\in \N$ is the number of ReLU gates in the ciruit, and $a^i \in \R^2$, $b_i, c_i \in \R$ for all $i=1, \ldots, w$. This implies that $f(x)$ is piecewise linear and the set of points where $f(x)$ is not differentiable is {\em precisely} the union of the $w$ lines $\langle a^i, x \rangle + b_i = 0$, $i=1, \ldots, w$. 
Since a union of lines cannot equal the union of the three half-lines $\{(x_1,x_2): x_1=x_2, x_1\geq 0\} \cup \{(0,x_2): x_2 \leq 0\} \cup \{(x_1,0): x_1 \leq 0\}$, we obtain the consequence that $\max\{0,x_1,x_2\}$ cannot be represented by a Sum-of-ReLU circuit, no matter how many ReLU gates are used.

\section{Simulating an LTF gate by a ReLU gate}\label{sec:LTF-ReLU}

\begin{claim}\label{LLbyLR}
Any LTF gate $\{-1,1\}^n \rightarrow \{-1,1\}$ can be simulated by a Sum-of-ReLU circuit with at most $2$ ReLU gates. 
\end{claim}
\begin{proof}
Given a LTF gate $(2{\mathbf 1}_{\langle a, x\rangle + b \geq 0}-1)$ it separates the points in $\{-1,1\}^n$ into two subsets such that the plane $\langle a, x\rangle + b = 0$ is a separating hyperplane between the two sets. Let $-p<0$ be the value of the function $\langle a, x\rangle + b$ at  that hypercube vertex on the ``-1'' side which is closest to this separating plane. Now imagine a continuous piecewise linear function $f : \mathbb{R} \rightarrow \mathbb{R}$ such that $f(x) =-1$ for $x \leq -p$, $f(x) = 1$ for $x \geq 0$ and for $x \in (-p,0)$ $f$ is the straight line function connecting $(-p,-1)$ to $(0,1)$. It follows from Corollary $3.1$ of our previous work, ~\cite{arora2016understanding} that this $f$ can be implemented by a $\mathbb{R} \rightarrow \mathbb{R}$ Sum-of-ReLU with at most $2$ ReLU gates hinged at the points $-p$ and $0$ on the domain. Because the affine transformation $\langle a, x\rangle + b$ can be implemented by the wires connecting the $n$ input nodes to the layer of ReLUs it follows that there exists a $\mathbb{R}^n \rightarrow \mathbb{R}$ Sum-of-ReLU with at most $2$ ReLU gates implementing the function $g(x) = f(\langle a, x\rangle + b) : \mathbb{R}^n \rightarrow \mathbb{R}$. Its clear that $g(\x) = \text{LTF}(\x)$ for all $\x \in \{-1,1\}^n$. 
\end{proof}

\section{PARITY on $k-$bits can be implemented by a $O(k)$ Sum-of-ReLU circuit}\label{sec:Parity-ReLU}

For this proof its convenient to think of the PARITY function as the following map, 
\begin{align*}\label{parity}
\text{PARITY} : \{0,1\}^k &\rightarrow \{0,1\}\\
\x &\mapsto \left (\sum_{i=1}^k x_i \right )\mod 2
\end{align*}
Its clear that that in the evaluation of the PARITY function as stated above the required sum over the coordinates of the input Boolean vector will take as value every integer in the set, $\{0,1,2,..,k\}$. The PARITY function can then be lifted to a $f :\R \rightarrow \R$ function such that, $f(y) = 0$ for all $y \leq 0$, $f(y) = y \mod 2$ for all $y \in {1,2,..,k}$, $f(y) = k \mod 2$ for all $y > k$ and for any $y \in (p,p+1)$ for $p \in \{0,1,..,k-1\}$ $f$ is the straight line function connecting the points, $(p, p \mod 2)$ and $(p +1, (p+1) \mod 2)$. Thus $f$ is a continuous piecewise linear function on $\R$ with $k+2$ linear pieces. Then it follows from Theorem $2.3$ of our previous work, ~\cite{arora2016understanding} that this $f$ can be implemented by a $\mathbb{R} \rightarrow \mathbb{R}$ Sum-of-ReLU circuit with at most $k+1$ ReLU gates hinged at the points $\{0,1,2,..,k\}$ on the domain. The wires from the $k$ inputs of the ReLU gates can implement the linear function $\sum_{i=1}^k x_i$. Thus it follows that there exists a $\R^k \rightarrow \R$ Sum-of-ReLU circuit (say C) such that, $\text{C}(\x) = \text{PARITY}(\x)$ for all $\x \in \{0,1\}^k$.  

\end{document}